\documentclass[12pt,final,
thmsa,
epsf,a4paper]
{article}
\usepackage{booktabs}

\newtheorem{theorem}{Theorem}

\newtheorem{corollary}{Corollary}
\newtheorem{definition}{Definition}
\newtheorem{example}{Example}
\newtheorem{lemma}{Lemma}
\newtheorem{proposition}{Proposition}
\newtheorem{remark}{Remark}

\newenvironment{proof}[1][Proof]{\emph{#1.} }{\  \hfill $\square $ \vspace{5 pt}}

\usepackage{natbib}
\usepackage{amssymb}
\usepackage[dvips]{graphics}
\usepackage{amsmath}

\usepackage{mathpazo}

\usepackage{enumerate}

\usepackage{amsfonts}

\usepackage{amssymb}

\usepackage{enumerate}

\usepackage{mathrsfs}

\usepackage[usenames]{color}

\usepackage{pgf,tikz}
\usetikzlibrary{arrows}

\usetikzlibrary{decorations.pathreplacing,angles,quotes}
\usetikzlibrary{arrows.meta}
\usetikzlibrary{arrows, decorations.markings}
\tikzset{myptr/.style={decoration={markings,mark=at position 1 with %
       {\arrow[scale=2,>=stealth]{>}}},postaction={decorate}}}

\usepackage{hyperref}
\hypersetup{
    colorlinks,
    citecolor=blue,
    filecolor=black,
    linkcolor=blue,
    urlcolor=blue
}

\usepackage[utf8]{inputenc}
\usepackage{amssymb, amsmath,geometry}
\usepackage{fancyhdr}

\usepackage{soul}

\usepackage{emptypage}

\newcommand*\samethanks[1][\value{footnote}]{\footnotemark[#1]}

\usepackage[T1]{fontenc}
\DeclareFontFamily{T1}{calligra}{}
\DeclareFontShape{T1}{calligra}{m}{n}{<->s*[1.44]callig15}{}
\DeclareMathAlphabet\mathcalligra   {T1}{calligra} {m} {n}

\begin{document}
\title{Coalitional stability under myopic expectations and externalities\thanks{We thank Adriana Piazza, Juan Pablo Torres-Martínez, and Matteo Triossi for their helpful comments. We acknowledge the financial support
from Agencia Nacional de Investigación y Desarrollo (ANID) through Fondecyt de Iniciación 11251916, from  UNSL through grants 032016, 030120, and 030320, from Consejo Nacional
de Investigaciones Cient\'{\i}ficas y T\'{e}cnicas (CONICET) through grant
PIP 112-200801-00655.}}
\author{Agust\'in G. Bonifacio\thanks{Instituto de Matem\'atica Aplicada San Luis (UNSL-CONICET), Departamento de Matem\'atica (Universidad Nacional de San Luis), Argentina. E-mails: abonifacio@unsl.edu.ar, paneme@unsl.edu.ar.}\and Mar\'ia Hayd\'ee Fonseca-Mairena\thanks{Corresponding author. Departamento de Econom\'ia y Administraci\'on, Universidad Cat\'olica del Maule. San Miguel 3605, Talca, Chile. Centro de Estudios Urbano Territoriales.  E-mail: mfonseca@ucm.cl.}\and Pablo Neme\samethanks[2]
}

\maketitle

\begin{abstract}  

We study coalition formation problems in the presence of externalities, focusing on settings where agents exhibit myopic expectations, that is, they evaluate potential deviations based solely on the immediate outcome, assuming no further reactions or reorganization by others. First, we establish a sufficient condition for the non-emptiness of both the core and the stable set. In the case of the core, our condition for ensuring non-emptiness also provides a characterization of all core partitions. We then turn our attention to problems with order-preserving preferences. Under our sufficient condition, the core and the stable set not only exist but also coincide, and convergence to a stable outcome is guaranteed. Furthermore, using the notion of absorbing set, we draw a connection between problems with order-preserving preferences and those without externalities. This allows us to lift known core non-emptiness results from the latter setting to the former, thereby establishing a novel bridge between these two classes of problems.

\medskip 
  
  


\noindent Economic Literature Classification Numbers: C71, C78, D62.


\end{abstract}

\section{Introduction}\label{introduccion}
We study coalition formation problems with externalities and myopic expectations. Externalities arise because agents care not only about the coalition they are in but also about on what coalitions other agents are. In other words, preferences are defined over partitions rather than individual coalitions.

Relevant examples include sports leagues, where players value both their team's internal quality and talent distribution across rival teams, as this affects their chances of success. Similarly, in political alliance formation, parties or candidates consider not only the direct benefits of their coalition but also the composition of other alliances, which may lead them to adjust their positions to counterbalance unfavorable political configurations. Likewise, in school choice, families and students evaluate not only the quality of a particular school or university but also the profile of students in other institutions, since a school's reputation may depend on both its performance and that of its peers.

Expectations are \emph{myopic} in the sense that, when evaluating a blocking move, agents focus solely on the immediate outcome, assuming that those outside the blocking coalition will not react by reorganizing themselves.\footnote{This notion corresponds to the \emph{$\gamma$-model} discussed by \cite{hart1983endogenous}. These authors argue that when a coalition is formed through the agreement of all its members but some agents later withdraw, the coalition dissolves, leaving the remaining agents as singletons. This notion is also studied by \cite{bloch2014expectation} and is named \emph{disintegration rule}.} The advantage of myopic blocking, or assuming myopic expectations, is that agents only need to consider a single scenario: the one that emerges immediately after the blocking occurs. In contrast, other types of expectations, such as prudent or optimistic, assume highly sophisticated agents who can anticipate specific responses from those not involved in the blocking decision \citep[see][among others]{bando2012many,bando2016two,ehlers2018strategy,fonseca2023coalition}.

This tension between simplicity of reasoning and the complexity of external responses has motivated recent research on how stability can be restored in the presence of externalities. For example, \cite{piazza2024coalitional} study matching markets with externalities and random preferences, and show that the likelihood of finding a coalitionally stable matching increases with market size, under mild conditions. Although their probabilistic perspective contrasts with our structural approach, both contributions underscore the importance of capturing externalities to define meaningful notions of stability.
Naturally, in the absence of externalities, expectations regarding how others might react to deviations become irrelevant.

Based on the myopic blocking notion, we first focus on two widely studied sets: the \emph{core} and the \emph{stable set}. Both consist of partitions in which no set of agents will block in order to benefit. In the case of core partitions, we only require that no set of agents benefit by deviating to a new partition in which they may reorganize into more than one new coalition. In contrast, for partitions in the stable set, we require that there be no set of deviating agents forming one coalition in the new partition. Therefore, it is straightforward from both definitions that the core is a subset of the stable set.

Although myopic expectations are highly relevant in these settings, a major challenge is that the set of stable outcomes, and therefore the core, is often empty. Therefore, it is common in the literature to study conditions under which stability is guaranteed. For example, in coalition formation problems without externalities, this has been explored in \cite{banerjee2001core, bogomolnaia2002stability, iehle2007core}, and in marriage problems with externalities in \cite{sasaki1996two, mumcu2010stable}. Likewise, we focus on identifying sufficient conditions that guarantee the existence of stable partitions in coalition formation problems with externalities, which encompass the aforementioned works. More precisely, we generalize the sufficient conditions proposed by \cite{banerjee2001core} by allowing for externalities and also extend the conditions established by \cite{mumcu2010stable} by allowing for more general coalitions than those formed by one man and one woman. To this end, we build on the intuitions behind the \emph{top-coalition property} introduced by \cite{banerjee2001core} 
 and the \emph{top-matching collection} introduced by \cite{mumcu2010stable}. 
Specifically, we introduce the notion of \emph{effective top-partition collection}. The intuition behind this condition is that agents share a common set of most preferred partitions (the top-partition collection) while still exhibiting enough diversity in how they rank them. That is, whenever there is more than one partition in the top-partition collection, for any pair of distinct partitions, we can identify a group of agents who are together as a coalition in one of them and unanimously prefer that partition over the other. We show that the existence of such a collection is a sufficient condition for the non-emptiness of the core (Theorem \ref{core=v} and Corollary \ref{corenonempty}). Furthermore, Theorem \ref{core=v} offers a characterization of all partitions in the core, derived
from the given top-partition collection.   Under a weaker condition, the \emph{weak effective top-partition collection}, we guarantee that the stable set is non-empty (Theorem \ref{teovsubses} and Corollary \ref{corestablenonempty}).

We then focus on problems with order-preserving preferences, where agents exhibit weak externalities in the sense that they primarily care about the coalition they belong to and only then, in a lexicographic manner, consider the composition of other coalitions. In this setting, the effective top-partition collection not only ensures the non-emptiness of the stable set but also implies that the stable set coincides with the core and consists of a single partition (Theorem \ref{core=estable}).

Myopic expectations also allow us to consider a dynamic process in which an unstable partition evolves as certain agents collectively decide to deviate and form a new coalition, while the abandoned agents become singletons in the new partition. Having in mind such a dynamic process, once stability is ensured, it is relevant to examine whether agents can always reach stability through their own interactions or if an arbitrator is required. A coalition formation problem with a non-empty stable set exhibits \textit{convergence to stability} if, for every non-stable partition, it is always possible to reach a stable partition through the dynamic process.\footnote{The concept of convergence to stability has been studied in various settings, including two-sided matching models (one-to-one and many-to-one) and one-sided matching models (such as coalition formation problems and roommate problems), both without externalities (see, e.g., \cite{roth1990random,chung2000existence,diamantoudi2004random,bonifacio2024characterization}) and with externalities (see \cite{dosSantosBraitt2021matching}, among others).} In general, convergence to stability is not guaranteed. However, in our setting with externalities, we show that in problems with order-preserving preferences, the effective top-partition collection condition is sufficient to guarantee convergence to stability (Theorem \ref{teoconver}).

Another widely studied notion is that of an \emph{absorbing set}, which is a minimal collection of partitions that, once reached through the above-mentioned dynamic process, will not be left.\footnote{This concept has been examined under different names and in various settings. To our knowledge, \cite{schwartz1970possibility} was the first to introduce it in the context of collective decision-making problems. See also \cite{jackson2002evolution,inarra2013absorbing,olaizola2014asymmetric}. The union of absorbing sets forms what is referred to as the \textit{admissible set} \citep{kalai1977admissible}. More recently, \cite{demuynck2019myopic} defined the \textit{myopic stable set} in a broad class of social environments and explored its relationship with other solution concepts. For general coalition formation problems without externalities, \cite{bonifacio2024characterization} characterize absorbing sets in terms of a special structure on preferences called the \textit{reduced form}.} The importance of the notion of an absorbing set is two-fold. On the one hand, an absorbing set is a more general notion than a stable partition, since a stable partition corresponds to an absorbing set consisting solely of that partition (trivial absorbing set), while non-trivial absorbing sets contain multiple partitions, none of which are stable. On the other hand, it provides a powerful tool for finding other sufficient conditions to ensure stability when the existence of an (weak) effective top-partition collection cannot be guaranteed. More precisely, with order-preserving preferences, we show that problems that induce the same problem without externalities generate the same absorbing sets (Theorem \ref{teo associated problems}). 
This finding is significant, as it shows that several sufficient conditions for the existence of a non-empty stable set---and, therefore, a non-empty core---in coalition formation problems with externalities can be derived by adapting results from classical settings without externalities.\footnote{For properties guaranteeing stability in coalition formation problems without externalities, see  \cite{farrell1988partnerships,banerjee2001core,bogomolnaia2002stability,alcalde2006,pycia2012stability,klaus2023core}, among others.} In particular, this theorem allows us to extend known results to environments with order-preserving preferences, providing a unified framework for understanding stability under more complex interactions (Corollary \ref{corollary:x-condition}).

Finally, in Section \ref{Seccion  Beyond myopic} we go beyond the standard myopic expectations by relaxing the assumption that non-participating agents dissolve into singletons after a deviation. We introduce a generalized notion of weak-myopic blocking, which allows for more flexible post-deviation reconfigurations and gives rise to the weak-myopic core. We also explore how different types of expectations---prudent and optimistic---affect stability outcomes, and show that, under order-preserving preferences, these refinements do not alter the collection of absorbing sets. This robustness result further supports the use of absorbing sets as a unifying framework for understanding stability under a range of behavioral assumptions.

The remainder of the paper is organized as follows. Section \ref{el modelo} presents the model and preliminaries. The main results are collected in Section \ref{seccion resultados}. Subsection \ref{subseccion de unrestricted domain} examines sufficient conditions to ensure the non-emptiness of the core and the stable set in problems with unrestricted preferences. Subsection \ref{subseccion de order-preserving preferences} focuses on problems with order-preserving preferences, strengthening the stability results and analyzing convergence to stability. Additional results based on the notion of absorbing set are presented in Subsection \ref{subseccion de absorbing sets}. Section \ref{Seccion Beyond myopic} explores how some of our results generalize beyond the assumption of myopic expectations. Finally, Section \ref{conclusiones} provides concluding remarks.


\section{The model}\label{el modelo}
Let $N=\{1,\ldots,n\}$ be a finite set of agents. Each non-empty $C\subseteq N$ is called a coalition. Let  $\mathcal{K}\subseteq 2^{N}\setminus \lbrace \emptyset \rbrace$ be the set of admissible coalitions. A \emph{partition} is a collection of disjoint coalitions $\pi=\{C_1,\ldots,C_k\}\subseteq \mathcal{K}$ such that $\cup_{i=1}^{k}C_i=N$.
Let $\Pi$ be the set of admissible partitions. 
Each agent $i\in N$ has complete and transitive preferences, denoted by  $\succsim_{i}$, over $\Pi$. Let $\succ_i$ and $\sim_i$ be the asymmetric and symmetric parts of $\succsim_i$, respectively. A preference profile is  $\succsim=\left(\succsim_{i}\right)_{i\in N}$. A set of permissible coalitions $\mathcal{K}$ and a preference profile $\succsim$ over admissible partitions defines a \textit{(coalition formation) problem}, denoted by $\left(\mathcal{K},\succsim\right)$.

The \textit{domain of problems with unrestricted preferences}, denoted by $\mathcal{R}$, consists of all problems with complete and transitive preferences, which also satisfy the condition that no agent can be indifferent between two structures that leave her in two different coalitions. Formally, for each $i\in N$ and each $\succsim_i$, we have that  $\pi \sim_i\pi'$ implies $ \pi(i)=\pi'(i)$.\footnote{This condition is commonly required in the literature to relate a coalition formation problem with externalities to one without externalities and with strict preferences \citep[see][for more details]{hong2022core,fonseca2023coalition,fonseca2023notes}.} 



We say that a problem $\left(\mathcal{K},\succsim\right)$ \textit{has externalities} if there is at least one agent who is not
indifferent between two partitions in which she belongs to the same coalition, i.e., if there is $i\in N$ such that $\pi \succ_i \pi'$ even though $\pi(i)=\pi'(i)$. This implies that the welfare of each agent also depends on the coalitions she does not belong to. 

    A partition $\pi$ is \textit{(myopically) blocked  by a coalition $C$ via a partition $\pi'$} if:
    \begin{enumerate}[(i)]
        \item $\pi' \succsim_i \pi$ for each $i\in C,$
        \item $\pi' \succ_j \pi$ for some $j\in C$,
        \item for each $C'\in \pi$ such that $C'\cap C\neq \emptyset$, $\pi'(j)=\{j\}$ for each $j\in C' \setminus C,$ and
        \item for each $C'\in \pi$ such that $C'\cap C= \emptyset$, $C' \in \pi'$.
    \end{enumerate}

Condition (i) establishes that no member of the coalition is worse off under the new partition compared to the original. Condition (ii) specifies that there is at least one member of the coalition who strictly prefers the alternative partition over the original. 
Condition (iii) states that the abandoned agents must stay as singletons in the new partition. This prevents those outside the coalition from remaining in a group with coalition members and potentially sharing in the benefits of the deviation. Finally, Condition (iv) guarantees that any group from the original partition that does not include any members of the coalition must remain unchanged in the alternative partition. This condition preserves the original partition for those not involved in the deviation. Together, these conditions reflect the essence of myopic blocking. The term ``myopic'' here captures the short-sighted nature of the blocking move: the coalition is concerned only with securing immediate improvements for its members, without taking into account any broader, possibly beneficial, reconfigurations that might occur if the deviation were allowed to affect all agents. Moreover, in such a myopic blocking process, it is assumed that agents outside the blocking coalition will not react strategically by forming new groupings in response to the deviation.\footnote{When externalities are present, agents outside the deviating coalition react can vary, leading to different notions of blocking. The myopic blocking approach assumes that outsiders do not react. Alternative concepts, such as prudent or optimistic blocking, account for specific forms of reactions. In the absence of externalities, however, all notions of blocking are equivalent, as agents outside the deviating coalition are unaffected and do not influence the feasibility of a deviation \citep[see][for more details]{sasaki1996two,bloch2014expectation}.}

\begin{definition}\label{bloqueo core-wise}
The set of partitions that are not blocked is called the \textbf{(myopic) core}, and denoted by $\mathcal{C}.$
\end{definition}

Note that, the notion of blocking assumes that the deviating coalition is not necessarily required to remain formed in the new partition $\pi'$. Instead, coalition members may reorganize into different groups after the deviation. If we impose that the deviating coalition must stay formed in the new partition, we can define a domination relation among partitions that is useful to define stability.

In problems without externalities, the core and the set of stable partitions coincide because individuals evaluate only their own coalition, without considering how others are grouped. 

However, in problems with externalities the relation between the core and the set of stable partitions is more subtle. A partition $\pi$ could be stable even though it is blocked by a coalition $C$ via a partition $\pi'$ \emph{as long as $C$ does not belong to $\pi'$}. Therefore,  stability involves a stronger notion of blocking because it demands that coalition $C$ must be formed in the new partition. Formally,

\begin{definition}\label{bloqueo estable}
    A partition $\pi$ is \textbf{(myopically) stable} if there is no partition $\pi'$ and coalition $C$ such that $\pi$ is blocked by coalition $C$ via $\pi'$ and $C\in \pi'$. The set of all stable partitions is denoted by $\mathcal{S}.$ 
\end{definition}

 This implies that,     
as a result, the core is a subset of the set of stable partitions.

\begin{remark}
By Definition \ref{bloqueo core-wise} and Definition \ref{bloqueo estable}, $\mathcal{C}\subseteq \mathcal{S}
$.    
\end{remark}

In problems with externalities, the stable set (and therefore the core) can be empty when agents are assumed to be myopic \citep{sasaki1996two} as shown in the following example.
\begin{example}\label{ejemplo 0}
 Consider the three-agent problem  with externalities given in Table \ref{table ejemplo de 3 agentes con estables vacio}.\footnote{Throughout the paper, in order to ease notation, we denote coalitions without curly brackets and commas, i.e.,  coalition $\{1,2\}$ is simply written $12$.}   
Observe that partition $\{13,2\}$ is blocked by coalition $23$, $\{12,3\}$ is blocked by coalition $13$, $\{1,23\}$ is blocked by coalition $12$ and, therefore, there is no stable partition. 

\begin{table}[ht!]{\small
    \centering
         \begin{tabular}{@{}ccccccccc@{}}\hline
        $\boldsymbol{1}$ && $\boldsymbol{2}$ && $\boldsymbol{3}$  \\ 
        \hline
        $\{13,2\}$ && $\{12,3\}$ && $\{1,23\}$ \\ 
        $\{12,3\}$ && $\{1,23\}$ && $\{13,2\}$  \\ 
        $\{1,23\}$&& $\{13,2\}$ && $\{12,3\}$ \\ 
       $\{1,2,3\}$  && $\{1,2,3\}$ && $\{1,2,3\}$  \\ \hline
         \end{tabular} 
         \caption{\small A problem with externalities and empty stable set.}
          \label{table ejemplo de 3 agentes con estables vacio}}
\end{table}
\end{example}

\section{Results}\label{seccion resultados}

In this section, following the ideas of \cite{banerjee2001core} for coalition formation problems without externalities and of \cite{mumcu2010stable} for the marriage problem with externalities, we define sufficient conditions for the non-emptiness of the core and the stable set in coalition formation problems with externalities.

The condition we impose to guarantee a non-empty stable set applies to the $k$-best partitions that all agents share.  
Note that such a $k$ always exists and is less than or equal to the total number of partitions in the problem, i.e., $k \leq |\Pi|.$  
Let $t_\ell(\succ_{i})$ denote the partition in the $\ell$-th position in preference $\succ_i$.  
We now formally define the set of partitions that all agents agree to be the $k$-best partitions. 

  \begin{definition}\label{defino top partition collection}
A non-empty set of partitions  $\mathcal{V}$ is a \textbf{top-partition collection} if $t_\ell(\succ_{i})\in \mathcal{V}$ for each $i\in N$ and for each $\ell\in\{1,\ldots,|\mathcal{V}|\}$.
  \end{definition}
Note that, given a problem with externalities, there may be more than one top-partition collection. However, we will focus on problems where agents agree on the $k$-best partitions but sufficiently disagree on how to rank them.   
Specifically, for any two partitions $\pi,\pi'\in \mathcal{V}$, there is a coalition in $\pi$ but not in $\pi'$ such that each agent in the coalition prefers $\pi$ over $\pi'$. To formalize this idea,  given $\pi,\pi'\in \Pi$, let $D(\pi,\pi')$ be the set of agents that belong to different coalitions under $\pi$ and $\pi'$, i.e., $D(\pi,\pi')=\{i\in N:\pi(i) \neq \pi'(i)\}$.  
\begin{definition}
    A top-partition collection $\mathcal{V}$ is \textbf{effective} if either $|\mathcal{V}|=1$ or whenever $|\mathcal{V}|>1$ we have that, for any two different partitions $\pi,\pi' \in \mathcal{V},$ there is a non-empty set $S \subseteq D(\pi,\pi')$ such that $S\in \pi$ and $\pi \succ_i \pi'$ for each $i\in S.$
\end{definition}

It is important to highlight that if there is an effective top-partition collection in a problem, it is unique.

\subsection{Problems with  unrestricted preferences}\label{subseccion de unrestricted domain}

The definition of an effective top-partition collection ensures that agents cannot form a blocking coalition to move from one top-ranked partition to another. By imposing this restriction, the definition reinforces the stability of the core and the stable set in coalition formation problems with externalities. Moreover, the existence of an effective top-partition collection not only guarantees the non-emptiness of the core, but also enables us to precisely identify the core partitions. The following theorem characterizes the core partitions in terms of an effective top-partition collection. 

\begin{theorem}\label{core=v}
    Let $\left(\mathcal{K},\succsim\right)$ be a 
    problem 
    in 
    $\mathcal{R}$. If there is an effective top-partition collection $\mathcal{V}$, then $\mathcal{C}=\mathcal{V}.$ 
\end{theorem}
\begin{proof}
    First we prove that $\mathcal{C}\subseteq\mathcal{V}.$ Let $\pi\notin \mathcal{V}$ and let $\pi' \in \mathcal{V}.$ Thus, $\pi$ is blocked by the grand coalition $N$ via $\pi'$. Then, $\pi\notin \mathcal{C}.$
    Second, we prove that $\mathcal{V}\subseteq\mathcal{C}.$ If $|\mathcal{V}|=1$, it is straightforward that $\mathcal{V}=\mathcal{C}.$ If $|\mathcal{V}|>1$, take $\pi\in \mathcal{V}$ and assume $\pi\notin \mathcal{C}.$ Thus, there are $\pi'\in \Pi \setminus \{\pi\}$ and $T\subseteq N$ such that $\pi$ is blocked by $T$ via $\pi'$. Then $\pi'\in \mathcal{V}$.  W.l.o.g., assume that $T\subseteq D(\pi',\pi).$ Then, $\pi' \succ_i \pi$ for each $i\in T.$ Additionally, there is $T'\subseteq T$ such that $T'\in \pi'$.    Moreover, since $\mathcal{V}$ is effective, there is a non-empty coalition  $S\subseteq D(\pi',\pi)$ such that $S\in \pi$ and $\pi \succ_i \pi'$ for each $i\in S.$ Then, $S\cap T=\emptyset.$ By Definition \ref{bloqueo core-wise} (iv), $\pi(i)=\pi'(i)$ for each $i\in S,$ contradicting that $S\subseteq D(\pi',\pi),$ implying that $\pi\in \mathcal{C}.$ Therefore, $\mathcal{V}=\mathcal{C}.$
\end{proof}

\begin{corollary}\label{corenonempty}
    Let $\left(\mathcal{K},\succsim\right)$ be a problem in
    $\mathcal{R}$. If there is an effective top-partition collection, then $\mathcal{C}\neq\emptyset.$ 
\end{corollary}

The following example demonstrates that the core can be non-empty even in the absence of an effective top-partition collection. Thus, the previous corollary provides a sufficient but not necessary condition.
\begin{example}\label{ejemplo_core_strictsubset_stable}
Consider the five-agent problem  with externalities given in Table \ref{table ejemplo de 5 agentes}. Observe that $\{12,34,5\}$ is in the core.  Moreover, partitions $\pi=\{13,24,5\}$ and $\pi'=\{12,34,5\}$ belong to any top-partition collection. However, there is not a set $S\subseteq D(\pi, \pi')$ such that $S\in \pi$ and $\pi\succ_i \pi'$ for each $i\in S$. Thus, there is no effective top-partition collection.  
\begin{table}[ht!]{\small
    \centering
         \begin{tabular}{@{}ccccccccccccc@{}}\hline
        $\boldsymbol{1}$ && $\boldsymbol{2}$ && $\boldsymbol{3}$ && $\boldsymbol{4}$ && $\boldsymbol{5}$ \\ 
        \hline
        $\{12,34,5\}$ && $\{12,34,5\}$ && $\{12,34,5\}$ && $\{12,34,5\}$ && $\{13,24,5\}$ \\ 
        $\{13,24,5\}$ && $\{13,24,5\}$ && $\{13,24,5\}$&& $\{13,24,5\}$ && $\{12,34,5\}$\\ 
        $\vdots$  && $\vdots$ && $\vdots$ && $\vdots$ && $\vdots$ \\       \hline
         \end{tabular} 
          \caption{\small A problem with no effective top-partition collection and non-empty core.}
          \label{table ejemplo de 5 agentes}}
\end{table}
        \end{example}

Since, by definition, the core is a subset of the stable set, it is possible to have a problem where the core is empty while the stable set is not. The following example illustrates this fact.
\begin{example}
   Consider a game where \( N = \{1, 2, 3, 4, 5, 6\} \), the set of feasible coalitions is given by
\[
\mathcal{K} = \{12, 14, 16, 23, 25, 34, 36, 45, 56, 1, 2, 3, 4, 5, 6\},
\]
and preferences are described in Table \ref{table ejemplo de 6 agents core vacio y stable no}.
\\
\begin{table}[ht!]{\small
    \centering
         \begin{tabular}{@{}ccccccccccccccc@{}}\hline
        $\boldsymbol{1}$ && $\boldsymbol{2}$ && $\boldsymbol{3}$ && $\boldsymbol{4}$ && $\boldsymbol{5}$&& $\boldsymbol{6}$ \\ 
        \hline
        $\{16,34,25\}$ && $\{12,34,56\}$ && $\{12,36,45\}$ && $\{16,23,45\}$ &&  $\{16,34,25\}$ &&  $\{16,23,45\}$\\ 
       $\{14,36,25\}$ && $\{14,36,25\}$ && $\{16,23,45\}$ && $\{16,34,25\}$ &&  $\{14,36,25\}$ &&  $\{16,34,25\}$\\
        $\{12,36,45\}$ && $\{12,36,45\}$ && $\{14,23,56\}$ && $\{14,36,25\}$ &&  $\{12,36,45\}$ &&  $\{14,36,25\}$\\
         $\{12,34,56\}$ && $\{12,45,3,6\}$ && $\{16,23,4,5\}$ && $\{12,36,45\}$ &&   $\{12,34,56\}$ &&  $\{12,36,45\}$\\
          $\{12,45,3,6\}$ && $\{16,23,45\}$ && $\{16,34,25\}$ && $\{12,34,56\}$ &&  $\{12,45,3,6\}$ &&  $\{12,34,56\}$\\
           $\{16,23,45\}$ && $\{14,23,56\}$ && $\{14,36,25\}$ && $\{14,23,56\}$ &&  $\{16,23,45\}$ && $\{14,23,56\}$\\
            $\{14,23,56\}$ && $\{16,23,4,5\}$ && $\{12,34,56\}$ && $\vdots$ &&   $\{14,23,56\}$ && $\vdots$\\
             $\vdots$ && $\{16,34,25\}$ && $\vdots$ && &&  $\vdots$  &&  \\
          && $\vdots$ &&   &&   &&  \\       \hline
         \end{tabular} 
          \caption{\small A problem with empty core and non-empty stable set.}
          \label{table ejemplo de 6 agents core vacio y stable no}}
\end{table}
It can be observed that the stable set consists of the partitions $\{12,34,56\}$, $\{12,36,45\}$, $\{14,23,56\}$, and $\{14,36,25\}$.  
However, the partition $\{12,34,56\}$ is blocked by the coalition $3456$ via the partition $\{12,36,45\}$;  
partition $\{12,36,45\}$ is blocked by the coalition $1245$ via the partition $\{14,36,25\}$;  
partition $\{14,23,56\}$ is blocked by the coalition $N$ via the partition $\{12,36,45\}$;  
and partition $\{14,36,25\}$ is blocked by the coalition $1346$ via the partition $\{16,34,25\}$. This implies that $\mathcal{C}=\emptyset.$
\end{example}

Next, we introduce a weaker sufficient condition that guarantees the existence of a non-empty stable set.

\begin{definition}\label{defino weak effective}
Let $\mathcal{V}$ be a top-partition collection. We say that $\mathcal{V}$ is \textbf{weak effective} if given any two different partitions $\pi,\pi'\in \mathcal{V}$ such that there is $C\in \pi'$ and $\pi'$ is defined as follows
    $$\pi'(i)=\begin{cases}
        C& \text{ if }i\in C\\
         \pi(i)&\text{ if }\pi(i) \cap C=\emptyset\\
        \{i\} & \text{ otherwise}
    \end{cases}$$
then there is $j\in C$ such that $\pi\succ_j\pi'.$
    \end{definition}

This construction of $\pi'$ from $\pi$ corresponds to forming a new coalition $C$ while breaking apart any original coalitions in $\pi$ that intersect with $C$.
The weak effective condition then ensures that such a deviation is not unanimously appealing to the agents involved: at least one agent in $C$ would prefer to remain under the structure of $\pi$.
Unlike the stronger effective condition, which requires that for any pair of distinct top-partitions there exists a full coalition in one that strictly prefers it over the other, weak effectiveness only requires the existence of an agent disagreeing in the specific case where the deviation involves introducing a new coalition $C$ and adjusting the coalitions in the original partition that overlap with it.
This weaker requirement is sufficient to ensure that for each partition in the top-partition collection, there is no coalition that blocks it via another partition from such a collection.

In Example \ref{ejemplo_core_strictsubset_stable} the top-partition collection $\mathcal{V}=\{\pi,\pi'\}$ is weak effective (but not effective). In fact $\mathcal{C}\subsetneq\mathcal{V}=\mathcal{S}$. The next result demonstrates that the weak effective top-partition collection condition is sufficient to guarantee the existence of a non-empty stable set.

\begin{theorem}\label{teovsubses}
   Let $\left(\mathcal{K},\succsim\right)$ be a problem in
    $\mathcal{R}$ and let $\mathcal{V}$ be a top-partition collection. If $\mathcal{V}$ is weak effective, then $\mathcal{V}\subseteq \mathcal{S}.$
\end{theorem}
\begin{proof}
 Given  $\pi\in \mathcal{V}$, by Definition \ref{defino top partition collection}, there is no  
 $C$ which blocks $\pi$ via some $\pi'\in \Pi\setminus \mathcal{V}$ such that $C\in \pi'$.   Note that if $|\mathcal{V}|=1,$ $\mathcal{V}\subseteq \mathcal{S}$.

  If $|\mathcal{V}|>1,$ take $\pi\in \mathcal{V}.$ If $\pi\notin \mathcal{S},$ there is $\pi'\in \Pi$ and $C\in \pi'$ such that $C$ blocks $\pi$ via $\pi'$.  Then, $\pi'\in \mathcal{V}.$ Thus, by Definition \ref{defino weak effective}, $\mathcal{V}$ in not weak effective, a contradiction.
\end{proof}
\begin{corollary}\label{corestablenonempty}
   Let $\left(\mathcal{K},\succsim\right)$ be a problem in
    $\mathcal{R}$ and let $\mathcal{V}$ be a top-partition collection. If $\mathcal{V}$ is weak effective, then $\mathcal{S}\neq \emptyset.$
\end{corollary}


The following example demonstrates that the stable set can be non-empty even in the absence of a weak effective top-partition collection. Thus, the previous corollary provides a sufficient but not a necessary condition.

\begin{example}\label{ejemplo 1}
Consider the four-agent problem  with externalities given in Table \ref{table ejemplo de 4 agentes OP}.  Observe that both $\{12,34\}$ and $\{13,24\}$ are stable partitions. Note that, for instance, partitions $\pi=\{12,3,4\}$ and $\pi'=\{12,34\}$ belong to any top-partition collection. Moreover, if we consider $C=\{34\}$, then $\pi$ and $\pi'$ are related as in Definition \ref{defino weak effective}, but there is no agent $j\in\{34\}$ such that $\pi \succ_j \pi'.$ Thus, there is no weak effective top-partition collection.
\begin{table}[ht!]{\small
    \centering
         \begin{tabular}{@{}ccccccccccc@{}}\hline
        $\boldsymbol{1}$ && $\boldsymbol{2}$ && $\boldsymbol{3}$ && $\boldsymbol{4}$  \\ 
        \hline
        $\{12,34\}$ && $\{13,24\}$ && $\{13,24\}$ && $\{12,34\}$ \\ 
        $\{12,3,4\}$ && $\{1,3,24\}$ && $\{13,2,4\}$ && $\{1,2,34\}$ \\ 
        $\{13,24\}$&& $\{12,34\}$ && $\{12,34\}$ &&$\{13,24\}$  \\ 
       $\{13,2,4\}$  && $\{12,3,4\}$ && $\{1,2,34\}$ &&$\{1,3,24\}$   \\
       $\vdots$  && $\vdots$ && $\vdots$ && $\vdots$ \\       \hline
         \end{tabular} 
         \caption{\small A problem with no weak effective top-partition collection and non-empty stable set.}
          \label{table ejemplo de 4 agentes OP}}
\end{table}

\end{example}

Once stability is guaranteed in coalition formation problems, it is relevant to study whether agents always reach stability when left to interact independently or whether an arbitrator must be introduced to ensure their stability. In order to do so, we first must establish a relation among partition.
The \textit{domination relation} $\gg$ over $\Pi$ is defined as follows: $\pi'\gg \pi$ if and only if there is a coalition $C\in \pi'$ such that partition $\pi$ is blocked by coalition $C$ via partition $\pi'$.\footnote{ It is important to highlight that the blocking coalition must be formed within the new partition. This provides the intuition that the domination relation is connected to the notion of stability. } Given the domination relation $\gg$ between partitions, let $\gg^T$ be the \emph{transitive closure} of $\gg$.  That is, given any two partitions $\pi$ and $\pi'$, we have that $\pi'\gg^T\pi$ if and only if there is a finite sequence of partitions $\pi=\pi^0,\pi^1,\ldots,\pi^J=\pi'$ such that, for each $j\in \{1,\ldots,J\}$, $\pi^{j}\gg \pi^{j-1}.$

Now, we are in a position to formally present the notion of convergence to stability.
\begin{definition}\label{defino convergencia}
    A problem $(\mathcal{K},\succsim)$ with non-empty stable set exhibits \textbf{convergence to stability} if for each non-stable partition $\pi'\in \Pi $ there is a stable partition $\pi\in \Pi $ such that $\pi\gg ^{T}\pi' $. 
\end{definition}

Even though the existence of an effective top-coalition collection (or its weaker version, a weak effective top-coalition collection) guarantees the non-emptiness of the stable set, the following result shows that there may be a lack of convergence to stability.

\begin{proposition}\label{remark no convergence}
For problems in $\mathcal{R}$, the existence of an effective top-partition collection does not guarantee convergence to stability. 
\end{proposition}
\begin{proof}
    We provide a counterexample to show that convergence to stability may fail even when an effective top-partition collection exists. Consider the five-agent problem with externalities given in Table \ref{table ejemplo de 5 agentes con externalidades} in which preferences can be completed arbitrarily.
\begin{table}[ht!]{\small
    \centering
         \begin{tabular}{@{}ccccccccc@{}}\hline
        $\boldsymbol{1}$ && $\boldsymbol{2}$ && $\boldsymbol{3}$ && $\boldsymbol{4}$ && $\boldsymbol{5}$ \\ 
        \hline
        $\{12,3,45\}$ && $\{12,3,45\}$ && $\{12,3,45\}$ && $\{12,3,45\}$ && $\{12,3,45\}$ \\ 
        $\{14,3,25\}$ && $\{123,4,5\}$ && $\{13,2,4,5\}$ && $\{123,4,5\}$ && $\{14,3,25\}$ \\ 
        $\{13,2,4,5\}$ && $\{14,3,25\}$ && $\{123,4,5\}$ && $\{14,2,3,5\}$ && $\{14,2,3,5\}$ \\ 
         $\{123,4,5\}$ && $\{14,2,3,5\}$ && $\{14,3,25\}$ && $\{13,2,4,5\}$ && $\vdots$\\ 
          $\{14,2,3,5\}$ && $\{13,2,4,5\}$ && $\{14,2,3,5\}$ && $\{14,3,25\}$ &&  \\ 
        $\vdots$  && $\vdots$ && $\vdots$ && $\vdots$ &&  \\       
        \hline
         \end{tabular} 
         \caption{\small A problem with an effective top-partition collection and lack of convergence to stability.}
          \label{table ejemplo de 5 agentes con externalidades}}
\end{table}
 In this problem, the effective top-partition collection is $\mathcal{V} = \{\{12, 3, 45\}\}$, which coincides with the stable set $\mathcal{S} = \{\{12, 3, 45\}\}$.  However, for any partition $\pi \neq \{12, 3, 45\}$, there exists no coalition $C$ such that $\{12, 3, 45\} \gg \pi$ via $C$. This implies that while $\{12, 3, 45\}$ is stable, other partitions cannot be directly dominated by it, preventing the dynamic process from converging to stability.   \end{proof}

\subsection{Problems with order-preserving preferences}\label{subseccion de order-preserving preferences}
Next, we study problems with order-preserving preferences. These preferences capture the idea that although an agent's well-being is affected by the coalitions formed by others (i.e., externalities are present), what matters most to the agent is the ranking of the coalitions to which she herself belongs. In other words, agents have preferences with externalities, but of a lexicographic nature: they primarily rank their own coalitions and consider the configuration of other coalitions only secondarily.

 Given a partition $\pi$ and an agent $i$, let $[\pi]_i$ denote the set of all partitions containing coalition $\pi(i)$, i.e. $[\pi]_i=\{\pi'\in \Pi:\pi'(i)=\pi(i)\}.$
\medskip

\noindent\textbf{Order-preserving preferences:}  a preference relation $\succsim_i$ is order-preserving if for each $\pi,\pi'\in \Pi$,
\begin{enumerate}[(i)]
    \item $\pi \sim_i\pi'$ implies $\pi(i)=\pi'(i)$,
    \item $\pi(i)\neq \pi'(i)$ and $\pi \succ_i \pi'$ imply $\widetilde{\pi}\succ_i\widetilde{\pi}'$ for each $\widetilde{\pi}\in [\pi]_i$ and each  $\widetilde{\pi}'\in [\pi']_i$.\footnote{In the literature, order-preserving preferences are also known as preferences with weak externalities or egocentric preferences \citep[see][among others]{sasaki1996two,hong2022core,fonseca2023notes}.}
\end{enumerate} Let $\boldsymbol{\mathcal{R}_{op}}$ denote the domain of problems with order-preserving preferences.

\medskip

Condition (i) says that an agent can be indifferent only between partitions that leave her in the same coalition. Condition (ii) says that given two different partitions, $\pi$ and $ \widetilde{\pi}$ such that  $ \widetilde{\pi} \in  [\pi]_i$, there is no partition $\pi' \notin  [\pi]_i$ such that $\pi \succ_i \pi' \succ_i \widetilde{\pi}$.  This means that each agent is concerned first of all about the coalition she belongs to, and then about the other coalitions.

Under this framework, effectiveness of the top-partition collection guarantees not only that the stable set is non-empty, but also that it coincides with the core and contains a unique partition. This fact will be shown in Theorem \ref{core=estable}, but in order to do so, we first present two technical lemmata.

The first shows that if the only difference between two partitions is that coalition $S$ is formed in one of them whereas subcoalitions of $S$ are formed in the other, then both partitions cannot be part of the same effective top-partition collection.   
 
\begin{lemma}\label{remark para que no sea efectivo}
       Let $\left(\mathcal{K},\succsim\right)$ be a problem in  $\mathcal{R}_{op}$, and let $\mathcal{V}$ be a top-partition collection. If $\pi,\pi'\in \mathcal{V}$ and $S\subseteq N$ are such that $S\in \pi$ and $S=D(\pi,\pi')$,  then  $\mathcal{V}$ is not effective.
\end{lemma}
\begin{proof}
Let $\mathcal{V}$ be a top-partition collection and assume that $\pi,\pi'\in \mathcal{V}$ and $S\subseteq N$ are such that $S\in \pi$ and $S=D(\pi,\pi')$. 
Then, $\pi'(i)\subsetneq S$ for each $i\in S$ and $\pi(i)=\pi'(i)$ for each $i\in N \setminus S.$ Assume that $\mathcal{V}$ is effective. Thus, $\pi\succ_i\pi'$ for each $i\in S$. Therefore, there is no $S'\subsetneq S$ such that $\pi' \succ_i \pi$ for each $i\in S'$. This contradicts that $\mathcal{V}$ is effective. 
\end{proof}


Notice that the previous result applies to problems with  unrestricted preferences. Lemma \ref{remark para que no sea efectivo} helps us to demonstrate that effective top-partition collections can contain only one element in problems with  order-preserving preferences. Formally,

\begin{lemma}\label{teo V no efectivo}
    Let $\left(\mathcal{K},\succsim\right)$ be a problem in $\mathcal{R}_{op}$. If $\mathcal{V}$ is a top-partition collection  with $|\mathcal{V}|>1$, then $\mathcal{V}$ is not effective.
\end{lemma}
\begin{proof}
    Let $\mathcal{V}$ be a top-partition collection with $|\mathcal{V}|>1$. Assume that $\mathcal{V}$ is effective. Let $\pi,\pi'\in \mathcal{V}$ be such that $\pi\neq \pi'$. Then, there is $S\subseteq D(\pi,\pi')$ such that $S\in \pi$ and $\pi \succ_i \pi'$ for each $i\in S$.
    Note that, for $i\in S$,  $\pi \succ_i \pi'$  and the fact that $\succ_i$ is an order-preserving preference imply that $\widetilde{\pi} \succ_i \widetilde{\pi}'$   for each $\widetilde{\pi} \in [\pi]_i$  and each $\widetilde{\pi}' \in [\pi']_i$. 
Then, $[\pi]_i\subseteq \mathcal{V}$. 
Take any $\pi^{\star}\in [\pi]_i$. There are two cases to consider:
\begin{enumerate}
\item[\textbf{1.}]\textbf{There is 
 a coalition $\boldsymbol{S'\in \pi^{\star}$ such that $S'\cap S=\emptyset$ and $|S'|>1.}$} Then, define $\pi^{\star\star}$ as follows:
$$\pi^{\star\star}(i)=\begin{cases}
    \{i\}& \text{for each }i\in S'\\
    \pi^{\star}(i)&\text{otherwise.}
\end{cases}$$
Note that $S'=D(\pi^\star,\pi^{\star\star})$. Thus, by Lemma \ref{remark para que no sea efectivo}, we have that  $\mathcal{V}$ is not effective, a contradiction. 

\item[\textbf{2.}]\textbf{There is no
 coalition $\boldsymbol{S'\in \pi^{\star}$ such that $S'\cap S=\emptyset$ and $|S'|>1.}$} Then, consider any partition $\pi^{\star\star}\in[\pi]_i$ such that there is a unique non-singleton coalition $S^\star$ with   $S^\star\in \pi^{\star\star}\setminus \pi^\star$. Note that,  $D(\pi^\star,\pi^{\star\star})=S^\star.$ Thus, by Lemma \ref{remark para que no sea efectivo}, we have that  $\mathcal{V}$ is not effective, a contradiction.
\end{enumerate}
Therefore, if $|\mathcal{V}|>1$, then $\mathcal{V}$ is not effective.
\end{proof}


Now, we are in a position to prove the following result.

\begin{theorem}\label{core=estable}
   Let $\left(\mathcal{K},\succsim\right)$ be a problem in $\mathcal{R}_{op}$. If the problem has an effective top-partition collection $\mathcal{V}$, then there is $\pi^\star\in \Pi$ such that $\{\pi^\star\} = \mathcal{V}=\mathcal{C}=\mathcal{S}.$ 
\end{theorem}
\begin{proof}
   Note that, by definition of  $\mathcal{C}$ and $\mathcal{S},$ we have that $\mathcal{C}\subseteq\mathcal{S}.$
   Assume  that $\mathcal{C}\subsetneq\mathcal{S}$. Then, there exists $\pi\in \mathcal{S}\setminus\mathcal{C}.$ Let $\mathcal{V}$ be an effective top-partition collection. By Lemma \ref{teo V no efectivo}, since $\mathcal{V}$ is effective, it follows that $|\mathcal{V}|=1$.  By Theorem \ref{core=v}, we have $\mathcal{V}=\mathcal{C}$. Since $\pi\notin \mathcal{C}$, it follows that $\pi\notin \mathcal{V}$. Let $\mathcal{V}=\{\pi^{\star}\}$. As $\pi^{\star}\neq \pi$, there exists $C\in \pi^{\star}\setminus \pi$. There are two cases to consider:
   \begin{enumerate}
       \item[\textbf{1.}] \textbf{$\boldsymbol{\pi^{\star}\gg \pi$ via $C}$.} Then, $\pi \notin \mathcal{S}$ which is a contradiction. 
        \item[\textbf{2.}] \textbf{$\boldsymbol{\pi^{\star}\not\gg \pi$ via $C}$.} Let $\widetilde{\pi}$ be such that 
       $$\widetilde{\pi}(i)=\begin{cases}
        C  &  i\in C \\
        \pi(i) & \pi(i)\cap C=\emptyset \\
        \{i\} & \text{otherwise}.
       \end{cases}
       $$

Let $j\in C$. Since $\mathcal{V}=\{\pi^{\star}\}$, we have $\widetilde{\pi}\in [\pi^{\star}]_j$. Moreover, because preferences are order-preserving, it follows that $\widetilde{\pi}\succ_j\pi$. Thus, $\widetilde{\pi}\gg \pi$ via $C$, which implies $\pi \notin \mathcal{S}$, again a contradiction. 
   \end{enumerate}

In both cases, we reach a contradiction, which implies that $\{\pi^\star\} = \mathcal{V}=\mathcal{C}=\mathcal{S}.$  
 \end{proof}

In the following example, we show that the converse of Theorem \ref{core=estable} does not hold. 

\begin{example}
 Consider the three-agent problem with externalities given in Table \ref{contraejemplo de vuelta teorema}. Observe that there is a unique stable partition, $\{123\}$, which is different from the top-partition collection. This implies that $|\mathcal{V}|>1$ and, therefore, $\mathcal{V}$ is not effective as Lemma \ref{teo V no efectivo} states.
\begin{table}[ht!]{\small
    \centering
         \begin{tabular}{@{}ccccccccc@{}}\hline
        $\boldsymbol{1}$ && $\boldsymbol{2}$ && $\boldsymbol{3}$  \\ 
        \hline
        $\{13,2\}$ && $\{12,3\}$ && $\{1,23\}$ \\ 
        $\{123\}$ && $\{123\}$ && $\{123\}$  \\ 
        $\{12,3\}$ && $\{1,23\}$ && $\{13,2\}$  \\ 
        $\{1,23\}$&& $\{13,2\}$ && $\{12,3\}$ \\ 
       $\{1,2,3\}$  && $\{1,2,3\}$ && $\{1,2,3\}$  \\ \hline
         \end{tabular} 
         \caption{\small A counter example for the converse of Theorem \ref{core=estable}.}
          \label{contraejemplo de vuelta teorema}}
\end{table}\hfill$\Diamond$
\end{example}
In contrast to the domain of problems with unrestricted preferences $\mathcal{R}$, where the existence of an effective top-partition collection does not guarantee convergence to stability (as shown in Proposition  \ref{remark no convergence}), in the domain of problems with order-preserving preferences $\mathcal{R}_{op}$, the existence of such a collection ensures convergence to stability.

\begin{theorem}\label{teoconver}
     Let $\left(\mathcal{K},\succsim\right)$ be a problem in  $\mathcal{R}_{op}$. If the problem has an effective top-partition collection, then there is convergence to stability.
\end{theorem}
\begin{proof}
    By Theorem \ref{core=v}, $\mathcal{V}=\mathcal{C}$. By Lemma \ref{teo V no efectivo}, $|\mathcal{C}|=1.$ By Theorem \ref{core=estable}, $\mathcal{S}=\mathcal{C}$ and thus $|\mathcal{S}|=1.$ Let $\pi^{\star}$ be the unique stable partition. Then, $t_1(\succ_i)=\pi^{\star}$ for each $i\in N$. Then, $\pi^{\star} \gg^T \pi$ for each $\pi\in \Pi.$ Therefore, there is convergence to stability.
\end{proof}

The following example shows that in a problem  with order-preserving preferences, the existence of an effective top-partition collection is not necessary to exhibit convergence to stability.
\medskip

\noindent\textbf{Example \ref{ejemplo 1} (continued)}
\emph{Consider the four-agent problem with externalities given in Table \ref{table ejemplo de 4 agentes OP} and assume that preferences are completed arbitrarily but subject to being order-preserving.  Recall that partitions $\{12,34\}$ and $\{13,24\}$ are stable. To observe that there is convergence to stability, we only need to verify that any unstable partition is $\gg^{T}$-- dominated by a stable one. Note that the only partitions that need to be analyzed are $\{12,3,4\}$, $\{13,2,4\}$, $\{1,3,24\}$, and $\{1,2,34\}$, as they are the only ones ranked above a stable partition for some agent. Thus, for these partitions, we have  $\{12,34\}\gg\{12,3,4\}$ via $34$, $\{13,24\}\gg\{13,2,4\}$ via $24$, $\{13,24\}\gg\{1,3,24\}$ via $13$, and $\{12,34\}\gg\{1,2,34\}$ via $12$. So in this example, there is convergence to stability. However, by Lemma \ref{teo V no efectivo}, there is no effective top-partition collection.}   \hfill$\Diamond$

\subsection{Further results with order-preserving preferences: absorbing sets}\label{subseccion de absorbing sets}

So far, we have developed our results based on a sufficient condition for the existence of a non-empty core (and, therefore, a non-empty stable set): the effective top-partition collection. As mentioned above, this condition is an adaptation for problems with externalities of the condition proposed by \cite{banerjee2001core}. However, in problems with externalities and order-preserving preferences, there is a more direct way to obtain sufficient conditions: by linking a problem with order-preserving preferences with one having preferences without externalities.

To do this, we first formally present the notion of preference without externalities as follows: \medskip

\noindent \textbf{Preferences without externalities:}  a preference relation $\succsim_i$ has no externalities if for each $\pi,\pi'\in \Pi$, $\pi \sim_i\pi'$ if and only if  $\pi(i)=\pi'(i)$. 

\noindent Let $\boldsymbol{\mathcal{R}_{we}}$ denote the domain of problems with preferences without externalities.
\medskip

In this restrictive domain of problems, an agent is only concerned with the coalitions she belongs to.  Note that, the domain of problems with preferences without externalities is a subset of the domain of problems with order-preserving preferences, i.e. $\mathcal{R}_{we}\subsetneq \mathcal{R}_{op}\subsetneq \mathcal{R}.$

 A broader solution concept useful to make precise the link between problems with order-preserving preferences and problems without externalities is  that of an ``absorbing set'', which can be described as follows.
Given the dynamic process $\gg$, an absorbing set is a minimal collection of partitions that, once entered through this dynamic process, is never left. Formally, 
\begin{definition}\label{definition of absorbing}
A non-empty set of partitions $\mathcal{A}\subseteq \Pi$ is an \textbf{absorbing set} whenever for each $\pi \in \mathcal{A}$ and each $\pi' \in \Pi \setminus \{\pi\},$ $$\pi' \gg^T \pi\text{ if and only if }\pi' \in \mathcal{A}.$$
\noindent If $|\mathcal{A}| \geq 3$,   $\mathcal{A}$ is said to be a \textbf{non-trivial absorbing set}. Otherwise, the absorbing set is \textbf{trivial}. 
\end{definition}\medskip
\noindent Notice that partitions in $\mathcal{A} $ are
symmetrically connected by the relation $\gg^{T}$, and that no partition in $\mathcal{A}$ is dominated by a partition outside the set.

An important feature of an absorbing set is that it always exists in coalition formation problems. Moreover, in the case that the stable set is non-empty, there is a one-to-one correspondence between the set of stable partitions and the set of trivial absorbing sets.

Now, we are in a position to formalize the link between problems with order-preserving preferences and problems without externalities.
\begin{definition}
    Given a problem $\left(\mathcal{K},\succsim\right) \in \mathcal{R}_{op}$, we can \textbf{induce} a problem  $\left(\mathcal{K},\succsim'\right) \in \mathcal{R}_{we}$ in the following way: for each agent $i\in N$ and two different partitions $\pi,\pi'\in \Pi$, 
    \begin{enumerate}[(i)]
        \item $\pi(i)=\pi'(i)$ implies that $\pi\sim'_i\pi'$, and
        \item $\pi(i)\neq \pi'(i)$ and $\pi\succ_i\pi'$ if and only if $\pi\succ'_i\pi'$.\footnote{The approach of relating problems with order-preserving preferences to corresponding settings without externalities has also been pursued in the literature, for instance, by \cite{sasaki1996two, fonseca2022incentives}. }
    \end{enumerate}  
\end{definition}

The following example shows a problem with order-preserving preferences that has two absorbing sets, one trivial and one non-trivial. It also shows how to induce a problem without externalities.
\begin{example}[Example 1 in \citealp{bonifacio2024characterization}]\label{ejemplo sobre induccion}
Let  $N=\lbrace
{1,2,3,4, 5}\rbrace$ and $$\mathcal{K}=\{12,123,15,1,23, 2, 34, 3, 45, 4, 5\}.$$ Consider the order-preserving preferences in Table \ref{table ejemplo con order preserving y varios Absorbing}. 
\begin{table}[ht!]{\small
    \centering
         \begin{tabular}{@{}ccccccccc@{}}\hline
        $\boldsymbol{1}$ && $\boldsymbol{2}$ && $\boldsymbol{3}$ && $\boldsymbol{4}$ && $\boldsymbol{5}$ \\ 
        \hline
        $\{12,3,4,5\}$ && $\{15,23,4\}$ && $\{12,34,5\}$ && $\{12,3,45\}$ && $\{15,23,4\}$ \\ 
        $\{12,34,5\}$ && $\{1,23,4,5\}$ && $\{15,2,34\}$ && $\{123,45\}$ && $\{15,2,34\}$ \\ 
        $\{12,3,45\}$ && $\{1,23,45\}$ && $\{1,2,34,5\}$ && $\{1,23,45\}$ && $\{15,2,3,4\}$ \\ 
         $\{123,4,5\}$ && $\{123,4,5\}$ && $\{123,4,5\}$ && $\{1,2,3,45\}$ && $\{12,3,45\}$\\ 
          $\{123,45\}$ && $\{123,45\}$ && $\{123,45\}$ && $\{12,34,5\}$ && $\{123,45\}$ \\ 
     $\{15,23,4\}$  &&  $\{12,3,4,5\}$    &&    $\{15,23,4\}$      &&    $\{15,2,34\}$  &&    $\{1,23,45\}$  \\
  $\{15,2,34\}$       &&   $\{12,3,45\}$      &&   $\{1,23,4,5\}$  && $\{1,2,34,5\}$     &&   $\{1,2,3,45\}$   \\
  $\{15,2,3,4\}$     &&   $\boldsymbol{\{12,34,5\}}$      &&   $ \{1,23,45\}$  && $\{12,3,4,5\}$     &&   $\{12,3,4,5\}$   \\
    $\{1,23,4,5\}$     &&   $\boldsymbol{\{15,2,34\}}$      &&   $\{12,3,4,5\} $  && $\{123,4,5\}$     &&   $\{12,34,5\}$   \\
   $\{1,23,45\}$     &&   $\{15,2,3,4\}$      &&   $\{12,3,45\} $  && $\{15,23,4\}$     &&   $\{123,4,5\}$   \\
    $\{1,2,34,5\}$     &&   $\{1,2,34,5\}$      &&   $\{15,2,3,4\} $  && $\{15,2,3,4\}$     &&   $\{1,23,4,5\}$   \\
   $ \{1,2,3,45\}$     &&   $\{1,2,3,45\}$      &&   $\{1,2,3,45\} $  && $\{1,23,4,5\}$     &&   $\{1,2,34,5\}$   \\
   $\{1,2,3,4,5\}$     &&   $\{1,2,3,4,5\}$      &&   $ \{1,2,3,4,5\}$  && $\{1,2,3,4,5\}$     &&   $\{1,2,3,4,5\}$   \\
        \hline
         \end{tabular} 
         \caption{\small A problem with order-preserving preferences and two absorbing sets.}
          \label{table ejemplo con order preserving y varios Absorbing}} 
\end{table}
In this example, there are two absorbing sets: one trivial and one non-trivial. The trivial absorbing set contains the stable partition $\{123,45\}$, while the non-trivial absorbing set consists of the following partitions: $\{12,34,5\}$, $\{12,3,45\}$, $\{15,23,5\}$, $\{15,2,34\}$, and $\{1,23,45\}.$
\begin{table}[ht!]
    \centering
        {\small \begin{tabular}{@{}c@{}}\hline
        $\boldsymbol{1}$  \\ 
        \hline
        $\{12,3,4,5\}, \{12,34,5\},\{12,3,45\}$ \\ 
         $\{123,4,5\},\{123,45\}$  \\ 
     $\{15,23,4\},\{15,2,34\},\{15,2,3,4\}$     \\
    $\{1,23,4,5\},\{1,23,45\},\{1,2,34,5\}, \{1,2,3,45\},\{1,2,3,4,5\}$ \\      
        \hline
         \end{tabular}}
         
         \caption{The induced preference without externalities of agent $1$ from Table \ref{table ejemplo con order preserving y varios Absorbing}.}
         \label{table del agente 1 inducida}
\end{table}
Since the preferences are order-preserving, we can induce a problem without externalities. Consider, for instance, the preferences of agent 1. The corresponding induced preferences without externalities are shown in Table \ref{table del agente 1 inducida}. It is worth mentioning, however, that the induced problem without externalities shares the same absorbing sets as the original problem with order-preserving preferences. This is because the induced problem is basically the same problem as in Example 1 in \cite{bonifacio2024characterization} since without externalities agents only care about the coalition they are in.
\hfill $\Diamond$
\end{example}

Note that a problem with order-preserving preferences induces a unique problem without externalities. However, a problem without externalities can be induced by many problems with order-preserving preferences. 
We say that two problems with order-preserving preferences are \emph{associated} if they induce the same problem without externalities. 
Note that a problem with order-preserving preferences is trivially associated with its induced problem without externalities.

In the previous example, we showed that the problem with order-preserving preferences shares the same absorbing sets as its induced counterpart. This is not a coincidence. In fact, the following theorem establishes that, under order-preserving preferences, associated problems have the same absorbing sets.

\begin{theorem}\label{teo associated problems}
Two associated problems with order-preserving preferences generate the same absorbing sets.
\end{theorem}
\begin{proof}
    Note that we only need to see that a problem with order-preserving preferences $(\mathcal{K},\succsim)$ and its induced problem without externalities $(\mathcal{K},\succsim')$ generate the same domination relation. Let $\gg$ denote the domination relation of the problem with order-preserving preferences and $\gg'$ denote the domination relation of the induced problem without externalities.    
    Let $\pi,\pi'\in \Pi$ such that $\pi' \gg \pi$. Then, there is $C\in \pi'$ such that (i) $\pi' \succsim_i \pi$ for each $i\in C,$ and (ii) $\pi' \succ_j \pi$ for some $j\in C$. Let $i\in C$. Since $\pi'(i)\neq \pi(i)$ and $\succsim_i$ is order-preserving, $\pi' \succ_i \pi$. Therefore $\pi' \succ'_i \pi$ and $\pi' \gg' \pi$. Similarly, we can see that $\pi' \gg' \pi$ implies  $\pi' \gg \pi$. Therefore, order $\gg$ and order $\gg'$ are equivalent.
\end{proof}

In the previous theorem, the assumption that the problems have order-preserving preferences is crucial for them to generate the same absorbing sets. In fact, this condition is quite sensitive. Consider, for example, a slight modification in the preferences of Example \ref{ejemplo sobre induccion}: assume a small change in the preference of agent $2$ by swapping the order of partitions $\{12,34,5\}$ and $\{15,2,34\}$ (these partitions are depicted in bold in Table \ref{table ejemplo con order preserving y varios Absorbing}). With this minor adjustment, the preferences are no longer order-preserving, and the new problem has three absorbing sets. A trivial absorbing set consisting of partition $\{123,45\}$, a new trivial absorbing set consisting of partition $\{15,2,34\}$, and a non-trivial absorbing set---different than the non-trivial absorbing set of the original problem---consisting of the following partitions: $\{12,3,4,5\}$, $\{123,4,5\}$, $\{15,2,3,4\}$, $\{1,23,4,5\}$, $\{1,2,34,5\}$, and $\{1,2,3,45\}$.

Theorem \ref{teo associated problems} is particularly powerful as it enables the identification of various sufficient conditions for the existence of a non-empty stable set---and, therefore, a non-empty core---in problems with order-preserving preferences, by  extending and applying established results from the literature in settings without externalities.

Let $\mathcal{R}_{op}^{\tau}$ denote the domain of problems with order-preserving preferences such that the corresponding induced problem without externalities satisfies the domain restriction $\tau$, where $\tau$ refers to one of the well-known conditions in the literature required to ensure stability on problems without externalities. Examples include: the common ranking property \citep{farrell1988partnerships}; the (weak) top-coalition property \citep{banerjee2001core}; the ordinally balanced property and the weak consecutiveness property \citep{bogomolnaia2002stability}; union responsiveness, intersection responsiveness, singularity, and essentiality \citep{alcalde2006}; the pairwise alignment property \citep{pycia2012stability}; and friend-oriented preferences \citep{klaus2023core}, among others.

Formally,
\[
\mathcal{R}_{op}^{\tau} = \left\{ (\mathcal{K},\succsim) \in \mathcal{R}_{op} \, \bigg| \, \text{the induced problem $ (\mathcal{K}, \succsim') \in  \mathcal{R}_{we}$ satisfies  restriction $\tau$} \right\}.
\]
\begin{corollary}\label{corollary:x-condition}
A problem   in  $\mathcal{R}_{op}^{\tau}$ has non-empty stable set.
\end{corollary}

\section{Beyond myopic expectations}\label{Seccion  Beyond myopic}
Until now, we have assumed that, in deciding to block, the deviating agents expect the coalitions they leave to break up into singletons. However, it is possible that the remaining agents regroup in various ways. To allow for this, we could reformulate the definition of myopic blocking by removing condition (iii), i.e. still rooted in myopic logic but beyond the strict singleton assumption, formally:

 A partition $\pi$ is \textit{weak-myopically blocked by a coalition $C$ via a partition $\pi'$} if:
    \begin{enumerate}[(i)]
        \item $\pi' \succsim_i \pi$ for each $i\in C,$
        \item $\pi' \succ_j \pi$ for some $j\in C$,
        \item for each $C'\in \pi$ such that $C'\cap C= \emptyset$, $C' \in \pi'$.
    \end{enumerate}

Note that this definition of blocking does not impose any condition on how the abandoned agents are reorganized. One special possibility is that the deviating agents might anticipate that the coalitions they leave behind will remain together (the $\delta$-\emph{model} in \cite{hart1983endogenous},  \emph{b-interchanges} in \cite{tamura1993transformation}, and \emph{projection rule} in \cite{bloch2014expectation}).

\begin{definition}\label{bloqueo2 core-wise}
The set of partitions that are not weak-myopically blocked is called the \textbf{weak myopic core}, and is denoted by $\mathcal{C}^{\omega}.$
\end{definition}

Note that the inclusion $\mathcal{C}^{\omega} \subseteq \mathcal{C}$ follows directly from the definition. However, the following example shows that the inclusion can be strict.

\begin{example}\label{ejemplo core star}
Consider the four-agent problem with externalities given in Table \ref{table ejemplo core star}.  We have the weak core is $\mathcal{C}^{\omega}
=\{\{13,24\}\}$ because partition $\{12,34\}$ is weak-myopically blocked by coalition $\{1,3\}$ via partition $\{13,24\}$ and all other partitions are blocked by the grand coalition. However, the myopic core is $\mathcal{C}=\{\{13,24\}, \{12,34\}\}$.

\begin{table}[ht!]{\small
    \centering
         \begin{tabular}{@{}ccccccccccc@{}}\hline
        $\boldsymbol{1}$ && $\boldsymbol{2}$ && $\boldsymbol{3}$ && $\boldsymbol{4}$  \\ 
        \hline
        $\{13,24\}$ && $\{12,34\}$ && $\{13,24\}$ && $\{12,34\}$ \\ 
        $\{12,34\}$ && $\{13,24\}$ && $\{12,34\}$ && $\{13,24\}$ \\ 
        $\{13,2,4\}$&& $\{13,2,4\}$ && $\{13,2,4\}$ &&$\{13,2,4\}$  \\ 
        $\vdots$  && $\vdots$ && $\vdots$ && $\vdots$ \\       \hline
         \end{tabular} 
         \caption{\small A problem with the weak myopic core strictly contained in the core.}
          \label{table ejemplo core star}}
\end{table}
\end{example}

When an effective top-partition collection exists, the two cores coincide and can be identified through the top-partition collection.

\begin{theorem}\label{core*=v}
    Let $\left(\mathcal{K},\succsim\right)$ be a 
    problem 
    in 
    $\mathcal{R}$. If there is an effective top-partition collection $\mathcal{V}$, then $\mathcal{C}^{\omega}=\mathcal{V}.$ 
\end{theorem}

The proof follows exactly the same steps as the proof of Theorem \ref{core=v}; therefore, it is omitted. \medskip

Moreover, when preferences are order-preserving, the expectations of the blocking coalition regarding the behavior of agents outside the coalition become irrelevant. For instance, two extreme cases of expectations can be considered. On one hand, the deviating agents might adopt a prudent stance, anticipating that external agents (including those who were not previously in any coalition with a deviating agent) will regroup in a way that seeks to harm the members of the deviating coalition as much as possible  \citep[\textit{pessimistic rule} in][]{bloch2014expectation}. On the other hand, the deviating agents might adopt an optimistic stance, expecting external agents to regroup in a way that seeks to benefit the members of the deviating coalition as much as possible \citep[\textit{optimistic rule} in][]{bloch2014expectation}.

Formally, let $\underline{\pi}$ denote the worst partition in the set $[\pi]_i$ for agent $i$, that is, $\pi \succsim_i \underline{\pi}$ for all $\pi \in [\pi]_i$, and let $\overline{\pi}$ denote the best partition in the set $[\pi]_i$ for agent $i$, that is, $\overline{\pi} \succsim_i \pi$ for all $\pi \in [\pi]_i$.

A partition $\pi$ is \textit{prudently blocked  by a coalition $C$ via a partition $\pi'$} if:
    \begin{enumerate}[(i)]
        \item $\underline{\pi}' \succsim_i \pi$ for each $i\in C,$
        \item $\underline{\pi}' \succ_j \pi$ for some $j\in C$,
    \end{enumerate}

A partition $\pi$ is \textit{optimistically blocked  by a coalition $C$ via a partition $\pi'$} if:
    \begin{enumerate}[(i)]
        \item $\overline{\pi}' \succsim_i \pi$ for each $i\in C,$
        \item $\overline{\pi}' \succ_j \pi$ for some $j\in C$,
    \end{enumerate}
    
The \textit{domination relation} $\underline{\gg}$ over $\Pi$ is defined as follows: $\pi'\underline{\gg} \pi$ if and only if there is a coalition $C\in \pi'$ such that partition $\pi$ is prudently blocked by coalition $C$ via partition $\pi'$. The \textit{domination relation} $\overline{\gg}$ over $\Pi$ is defined as follows: $\pi'\overline{\gg} \pi$ if and only if there is a coalition $C\in \pi'$ such that partition $\pi$ is optimistically blocked by coalition $C$ via partition $\pi'$. For each of these domination relations we can define the corresponding absorbing set following the lines of Definition \ref{definition of absorbing}, we can call them \textit{prudent absorbing set} and \textit{optimistic absorbing set}, respectively.

For each of these new notions of absorbing sets, an analogous result to Theorem \ref{teo associated problems} can be established.

\begin{theorem}\label{teo associated problems2}
Two associated problems with order-preserving preferences generate the same collection of prudent absorbing sets. Moreover, this collection is also equal to the collection of optimistic absorbing sets of the two associated problems. 
\end{theorem}

The proof proceeds in exactly the same way as that of Theorem \ref{teo associated problems}, and is therefore omitted.

\cite{fonseca2023notes} show that, in the marriage market with externalities and order-preserving preferences, the optimistic stable set coincides with the prudent stable set \citep[Theorem 1 in][] {fonseca2023notes}; moreover, both coincide with the stable set in the absence of externalities \citep[Theorem 2 in][]{fonseca2023notes}. Therefore, our Theorem \ref{teo associated problems2} generalizes their results, as it does not restrict attention to coalitions of at most two agents in two-sided problems and adopts a broader notion of stability.

\section{Conclusions}\label{conclusiones}

This paper addresses coalition formation problems with externalities and myopic expectations, offering new insights into the existence and dynamics of stable outcomes. We introduced the concepts of the effective and weak effective top-partition collections as sufficient conditions for guaranteeing the non-emptiness of the core and the stable set, respectively.
The effectiveness of the top-partition collection also serves as a sufficient condition for the non-emptiness of a weaker version of the myopic core, under expectations weaker than the myopic assumption.

Under order-preserving preferences, we also showed that the effective top-partition collection not only ensure the existence but also the uniqueness of the stable partition, which coincides with the core. 
Moreover, we examine a dynamic process of coalition formation and show that 
the effective top-partition collection condition ensures convergence to stability. This last fact highlights the compatibility between stability and decentralized adjustment dynamics in environments with externalities. Furthermore, we explored the role of absorbing sets as a more general solution concept. When the aforementioned sufficient conditions are not met, absorbing sets provide an alternative framework to study stability. In particular, we showed that, under order-preserving preferences, problems with externalities that induce the same underlying problem without externalities generate identical absorbing sets, regardless of the type of expectations. This connection allows for the transfer of known results from simpler settings to more complex environments.

Overall, our findings contribute to a deeper understanding of stability in coalition formation with externalities and open the door to future work exploring richer dynamic models, heterogeneous expectations, or strategic behavior in environments with interdependent preferences.

\bibliographystyle{ecta}
\bibliography{bibliografia}

\begin{thebibliography}{31}
\newcommand{\enquote}[1]{``#1''}
\expandafter\ifx\csname natexlab\endcsname\relax\def\natexlab#1{#1}\fi

\bibitem[\protect\citeauthoryear{Alcalde and Romero-Medina}{Alcalde and Romero-Medina}{2006}]{alcalde2006}
\textsc{Alcalde, J. and A.~Romero-Medina} (2006): \enquote{Coalition formation and stability,} \emph{Social Choice and Welfare}, 27, 365--375.

\bibitem[\protect\citeauthoryear{Bando}{Bando}{2012}]{bando2012many}
\textsc{Bando, K.} (2012): \enquote{Many-to-one matching markets with externalities among firms,} \emph{Journal of Mathematical Economics}, 48, 14--20.

\bibitem[\protect\citeauthoryear{Bando, Kawasaki, and Muto}{Bando et~al.}{2016}]{bando2016two}
\textsc{Bando, K., R.~Kawasaki, and S.~Muto} (2016): \enquote{Two-sided matching with externalities: A survey,} \emph{Journal of the Operations Research Society of Japan}, 59, 35--71.

\bibitem[\protect\citeauthoryear{Banerjee, Konishi, and S{\"o}nmez}{Banerjee et~al.}{2001}]{banerjee2001core}
\textsc{Banerjee, S., H.~Konishi, and T.~S{\"o}nmez} (2001): \enquote{Core in a simple coalition formation game,} \emph{Social Choice and Welfare}, 18, 135--153.

\bibitem[\protect\citeauthoryear{Bloch and Van~den Nouweland}{Bloch and Van~den Nouweland}{2014}]{bloch2014expectation}
\textsc{Bloch, F. and A.~Van~den Nouweland} (2014): \enquote{Expectation formation rules and the core of partition function games,} \emph{Games and Economic Behavior}, 88, 339--353.

\bibitem[\protect\citeauthoryear{Bogomolnaia and Jackson}{Bogomolnaia and Jackson}{2002}]{bogomolnaia2002stability}
\textsc{Bogomolnaia, A. and M.~O. Jackson} (2002): \enquote{The stability of hedonic coalition structures,} \emph{Games and Economic Behavior}, 38, 201--230.

\bibitem[\protect\citeauthoryear{Bonifacio, Inarra, and Neme}{Bonifacio et~al.}{2024}]{bonifacio2024characterization}
\textsc{Bonifacio, A., E.~Inarra, and P.~Neme} (2024): \enquote{A characterization of absorbing sets in coalition formation games,} \emph{Games and Economic Behavior}, 148, 1--22.

\bibitem[\protect\citeauthoryear{Chung}{Chung}{2000}]{chung2000existence}
\textsc{Chung, K.-S.} (2000): \enquote{On the existence of stable roommate matchings,} \emph{Games and Economic Behavior}, 33, 206--230.

\bibitem[\protect\citeauthoryear{Demuynck, Herings, Saulle, and Seel}{Demuynck et~al.}{2019}]{demuynck2019myopic}
\textsc{Demuynck, T., P.~J.-J. Herings, R.~D. Saulle, and C.~Seel} (2019): \enquote{The myopic stable set for social environments,} \emph{Econometrica}, 87, 111--138.

\bibitem[\protect\citeauthoryear{Diamantoudi, Miyagawa, and Xue}{Diamantoudi et~al.}{2004}]{diamantoudi2004random}
\textsc{Diamantoudi, E., E.~Miyagawa, and L.~Xue} (2004): \enquote{Random paths to stability in the roommate problem,} \emph{Games and Economic Behavior}, 48, 18--28.

\bibitem[\protect\citeauthoryear{dos Santos~Braitt and Torres-Martínez}{dos Santos~Braitt and Torres-Martínez}{2021}]{dosSantosBraitt2021matching}
\textsc{dos Santos~Braitt, M. and J.~P. Torres-Martínez} (2021): \enquote{Matching with externalities: The role of prudence and social connectedness in stability,} \emph{Journal of Mathematical Economics}, 92, 95--102.

\bibitem[\protect\citeauthoryear{Ehlers}{Ehlers}{2018}]{ehlers2018strategy}
\textsc{Ehlers, L.} (2018): \enquote{Strategy-proofness and essentially single-valued cores revisited,} \emph{Journal of Economic Theory}, 176, 393--407.

\bibitem[\protect\citeauthoryear{Farrell and Scotchmer}{Farrell and Scotchmer}{1988}]{farrell1988partnerships}
\textsc{Farrell, J. and S.~Scotchmer} (1988): \enquote{Partnerships,} \emph{The Quarterly Journal of Economics}, 103, 279--297.

\bibitem[\protect\citeauthoryear{Fonseca-Mairena and Triossi}{Fonseca-Mairena and Triossi}{2022}]{fonseca2022incentives}
\textsc{Fonseca-Mairena, M.~H. and M.~Triossi} (2022): \enquote{Incentives and implementation in allocation problems with externalities,} \emph{Journal of Mathematical Economics}, 99, 102613.

\bibitem[\protect\citeauthoryear{Fonseca-Mairena and Triossi}{Fonseca-Mairena and Triossi}{2023{\natexlab{a}}}]{fonseca2023coalition}
---\hspace{-.1pt}---\hspace{-.1pt}--- (2023{\natexlab{a}}): \enquote{Coalition formation problems with externalities,} \emph{Economics Letters}, 226, 111112.

\bibitem[\protect\citeauthoryear{Fonseca-Mairena and Triossi}{Fonseca-Mairena and Triossi}{2023{\natexlab{b}}}]{fonseca2023notes}
---\hspace{-.1pt}---\hspace{-.1pt}--- (2023{\natexlab{b}}): \enquote{Notes on marriage markets with weak externalities,} \emph{Bulletin of Economic Research}, 75, 860--868.

\bibitem[\protect\citeauthoryear{Hart and Kurz}{Hart and Kurz}{1983}]{hart1983endogenous}
\textsc{Hart, S. and M.~Kurz} (1983): \enquote{Endogenous formation of coalitions,} \emph{Econometrica}, 1047--1064.

\bibitem[\protect\citeauthoryear{Hong and Park}{Hong and Park}{2022}]{hong2022core}
\textsc{Hong, M. and J.~Park} (2022): \enquote{Core and top trading cycles in a market with indivisible goods and externalities,} \emph{Journal of Mathematical Economics}, 100, 102627.

\bibitem[\protect\citeauthoryear{Iehl{\'e}}{Iehl{\'e}}{2007}]{iehle2007core}
\textsc{Iehl{\'e}, V.} (2007): \enquote{The core-partition of a hedonic game,} \emph{Mathematical Social Sciences}, 54, 176--185.

\bibitem[\protect\citeauthoryear{Inarra, Larrea, and Molis}{Inarra et~al.}{2013}]{inarra2013absorbing}
\textsc{Inarra, E., C.~Larrea, and E.~Molis} (2013): \enquote{Absorbing sets in roommate problems,} \emph{Games and Economic Behavior}, 81, 165--178.

\bibitem[\protect\citeauthoryear{Jackson and Watts}{Jackson and Watts}{2002}]{jackson2002evolution}
\textsc{Jackson, M.~O. and A.~Watts} (2002): \enquote{The evolution of social and economic networks,} \emph{Journal of Economic Theory}, 106, 265--295.

\bibitem[\protect\citeauthoryear{Kalai and Schmeidler}{Kalai and Schmeidler}{1977}]{kalai1977admissible}
\textsc{Kalai, E. and D.~Schmeidler} (1977): \enquote{An admissible set occurring in various bargaining situations,} \emph{Journal of Economic Theory}, 14, 402--411.

\bibitem[\protect\citeauthoryear{Klaus, Klijn, and {\"O}zbilen}{Klaus et~al.}{2023}]{klaus2023core}
\textsc{Klaus, B., F.~Klijn, and S.~{\"O}zbilen} (2023): \enquote{Core Stability and Strategy-Proofness in Hedonic Coalition Formation Problems with Friend-Oriented Preferences,} \emph{Barcelona School of Economics Working Paper 1399}.

\bibitem[\protect\citeauthoryear{Mumcu and Saglam}{Mumcu and Saglam}{2010}]{mumcu2010stable}
\textsc{Mumcu, A. and I.~Saglam} (2010): \enquote{Stable one-to-one matchings with externalities,} \emph{Mathematical Social Sciences}, 60, 154--159.

\bibitem[\protect\citeauthoryear{Olaizola and Valenciano}{Olaizola and Valenciano}{2014}]{olaizola2014asymmetric}
\textsc{Olaizola, N. and F.~Valenciano} (2014): \enquote{Asymmetric flow networks,} \emph{European Journal of Operational Research}, 237, 566--579.

\bibitem[\protect\citeauthoryear{Piazza and Torres-Mart{\'\i}nez}{Piazza and Torres-Mart{\'\i}nez}{2024}]{piazza2024coalitional}
\textsc{Piazza, A. and J.~P. Torres-Mart{\'\i}nez} (2024): \enquote{Coalitional stability in matching problems with externalities and random preferences,} \emph{Games and Economic Behavior}, 143, 321--339.

\bibitem[\protect\citeauthoryear{Pycia}{Pycia}{2012}]{pycia2012stability}
\textsc{Pycia, M.} (2012): \enquote{Stability and preference alignment in matching and coalition formation,} \emph{Econometrica}, 80, 323--362.

\bibitem[\protect\citeauthoryear{Roth and Vande~Vate}{Roth and Vande~Vate}{1990}]{roth1990random}
\textsc{Roth, A.~E. and J.~H. Vande~Vate} (1990): \enquote{Random paths to stability in two-sided matching,} \emph{Econometrica}, 1475--1480.

\bibitem[\protect\citeauthoryear{Sasaki and Toda}{Sasaki and Toda}{1996}]{sasaki1996two}
\textsc{Sasaki, H. and M.~Toda} (1996): \enquote{Two-sided matching problems with externalities,} \emph{Journal of Economic Theory}, 70, 93--108.

\bibitem[\protect\citeauthoryear{Schwartz}{Schwartz}{1970}]{schwartz1970possibility}
\textsc{Schwartz, T.} (1970): \enquote{On the possibility of rational policy evaluation,} \emph{Theory and Decision}, 1, 89--106.

\bibitem[\protect\citeauthoryear{Tamura}{Tamura}{1993}]{tamura1993transformation}
\textsc{Tamura, A.} (1993): \enquote{Transformation from arbitrary matchings to stable matchings,} \emph{Journal of Combinatorial Theory, Series A}, 62, 310--323.

\end{thebibliography}

\end{document}